\documentclass[11pt]{article}
\usepackage{graphicx}
\usepackage{times}
\usepackage{latexsym}
\usepackage{amssymb}
\usepackage{citesort}
\usepackage{fullpage}
\usepackage{color}
\usepackage{verbatim}
\usepackage{epsfig}
\usepackage{enumerate}
\usepackage{url}

\newcommand{\calL}{\ensuremath{{\cal L}}}
\newcommand{\calR}{\ensuremath{{\cal R}}}

\newtheorem{theorem}{Theorem}[section]
\newtheorem{lemma}{Lemma}[section]

\usepackage{times}
\usepackage{latexsym}
\newenvironment{proof}{{\bf Proof:} }{\hspace*{\fill}$\Box$\vspace{2mm}}

\newtheorem{conjecture}{Conjecture}[section]

\begin{document}
\addtocounter{page}{-1}

\title{\LARGE\bf Computing Cartograms with Optimal Complexity}

\author{
      \large Md.~Jawaherul~Alam$^1$\\[+1mm]
%      \sf    Department of Computer Science,\\
      \sf    University of Arizona,\\
      \sf    Tucson, AZ, USA\\
        {\tt mjalam@email.arizona.edu}
\and
      \large Therese~Biedl$^2$\\[+1mm]
%      \sf    School of Computer Science,\\
      \sf    University of Waterloo,\\
      \sf    Waterloo, ON N2L 3G1, Canada\\
        {\tt biedl@uwaterloo.ca}
\and
      \large Stefan Felsner$^3$\\[+1mm]
%      \sf    Institut f\"ur Mathematik,\\
      \sf    Technische Universit\"at Berlin,\\
      \sf    Berlin, Germany\\
        {\tt felsner@math.tu-berlin.de}
\and
      \large Michael~Kaufmann\\[+1mm]
%      \sf    Wilhelm-Schickhard-Institut f\"ur Informatik,\\
      \sf    Eberhard Karls Universit\"at T\"ubingen,\\
      \sf    T\"ubingen, Germany\\
        {\tt mk@informatik.uni-tuebingen.de}
\and
      \large Stephen~G.~Kobourov$^1$\\[+1mm]
%      \sf    Department of Computer Science,\\
      \sf    University of Arizona,\\
      \sf    Tucson, AZ, USA\\
        {\tt kobourov@cs.arizona.edu}
\and
      \large Torsten~Ueckerdt$^3$\\[+1mm]
%      \sf    Institut f\"ur Mathematik,\\
      \sf    Technische Universit\"at Berlin,\\
      \sf    Berlin, Germany\\
        {\tt ueckerdt@math.tu-berlin.de}
}
\date{}

\maketitle
\let\thefootnote\relax\footnotetext{

\noindent$^*$ This research was initiated at the Dagstuhl Seminar 10461 on Schematization.

\noindent$^1$ Research funded in part by NSF grants CCF-0545743 and CCF-1115971.

\noindent$^2$ Research supported by NSERC.

\noindent$^3$ Research partially supported by EUROGIGA project GraDR and DFG Fe 340/7-2.
}
\begin{abstract}
\thispagestyle{empty}
In a rectilinear dual of a planar graph vertices are
represented by simple rectilinear polygons and edges are
represented by side-contact between the corresponding polygons.
A rectilinear dual is called a cartogram if the area of each region is
equal to a pre-specified weight of the corresponding vertex. The
complexity of a cartogram is determined by the maximum number of
corners (or sides) required for any polygon. In a series of papers the
polygonal complexity of such representations for maximal planar graphs
has been reduced from the initial 40 to 34, then to 12 and very
recently to the currently best known 10.  Here we describe a
construction with 8-sided polygons, which is optimal in
terms of polygonal complexity as 8-sided polygons are sometimes
necessary. Specifically, we show how to compute the combinatorial
structure and how to refine the representation into an area-universal
rectangular layout in linear time. The exact cartogram can be computed
from the area-universal rectangular layout with numerical iteration,
or can be approximated with a hill-climbing heuristic.

We also describe  an alternative construction for Hamiltonian maximal
planar graphs, which allows us to directly compute the cartograms in
linear time. Moreover, we prove that even for Hamiltonian graphs 8-sided
rectilinear polygons are necessary, by constructing a non-trivial
lower bound example. The complexity of the cartograms can be reduced
to 6 if the Hamiltonian path has the extra property that it is
one-legged, as in outer-planar graphs. Thus, we have optimal
representations (in terms of both polygonal complexity and running
time) for Hamiltonian maximal planar and maximal outer-planar graphs.
\end{abstract}

\newpage

\section{Introduction}
\label{sec:intro}

There is a large body of work about representing planar graphs as
contact graphs, i.e., graphs whose vertices are represented by
geometrical objects with edges corresponding to two objects touching
in some specified fashion. Typical classes of objects might be curves,
line segments,  or polygons. An early result is Koebe's 1936 theorem~\cite{Koebe36} that all planar graphs can be represented by
 touching disks.

In this paper, we consider contact representations of planar graphs, with vertices represented by simple interior-disjoint polygons
 and adjacencies represented by a non-trivial contact (shared boundary)
between the
 corresponding polygons. We are specifically interested in the
 rectilinear weighted version where the vertices are represented by
 simple (axis-aligned) rectilinear polygons.
This type of a representation is known
 as a {\em rectilinear dual} of the input planar graph.

 In the weighted version, the input is a planar graph $G=(V,E)$ along with a weight function $w:V(G)\rightarrow \mathbb{R}^+$ that
 assigns a weight to each vertex of $G$. A rectilinear dual is called a \textit{cartogram} if the area
 of each region is equal to the pre-specified weight of the
 corresponding vertex.  Such representations have practical applications in
 cartography~\cite{r-rsc-34}, geography~\cite{Tobler04thirtyfive} and sociology~\cite{HK98}, but also in VLSI Layout and
 floor-planning~\cite{MCP2002}.
Other applications can be found in visualization of relational data,
where using the adjacency of regions to represent edges in a graph can lead
 to a more compelling visualization than just drawing a line segment
 between two points~\cite{Buchsbaum08}.

 For rectilinear duals (unweighted) and for cartograms (weighted) it
 is often desirable, for aesthetic, practical and cognitive reasons,
 to limit the {\em polygonal complexity} of the representation,
 measured by the number of sides (or by the number of corners).
 Similarly, it is also desirable to minimize the unused area in the
 representation, also known as {\em holes} in floor-planning and VLSI
 layouts. A given rectilinear dual is {\em area-universal} if it can
 realize a cartogram with any pre-specified set of weights for the
 vertices of the graph without disturbing the underlying adjacencies
 and without increasing the polygonal complexity.

With these considerations in mind, we study the problem of
constructing area-universal rectilinear duals and show how to compute
cartograms with worst-case optimal polygonal complexity and without
any holes.

\subsection{Related Work}

In our paper and in most of the other papers cited here, ``planar graph''
refers to an inner-triangulated planar graph with a simple outer-face;
the former restriction is required if at most three rectilinear polygons are
allowed to meet in a point and the latter restriction is customary to
achieve that the union of all the polygons in the representation is a
rectangle.

Rectilinear duals (unweighted) were first studied in graph
theoretic context, and then with renewed interest in the context of
VLSI layouts and floor planning. It is known that $8$ sides are sometimes
necessary and always sufficient~\cite{He99,Liao03,ys93}.

The case when the rectilinear
polygons are restricted to rectangles has been of particular interest and
there are several (independent) characterizations of the class of
planar graphs that allows such {\em rectangular duals}~\cite{u-drg-53,LL-abrf-84,kk-rdpg-85}.
A historical overview
and a summary of the state of the art in the rectangle contact graphs
literature can be found in Buchsbaum {\em et al.}~\cite{Buchsbaum08}.

In the above results on rectilinear duals and rectangular duals, the areas of
the polygons are not considered; that is, these
results deal with the unweighted version of the problem.
The weighted version dates back to 1934 when Raisz described rectangular
cartograms~\cite{r-rsc-34}.
Algorithms by van Kreveld and Speckmann~\cite{ks07} and
Heilmann {\em et al.}~\cite{hkps04} yield representations with
touching rectangles but the adjacencies may be disturbed and
there may also be a small distortions of the weights. Recently, Eppstein
{\em et al.}~\cite{EMVS} characterized the class of planar graphs that
have area-universal rectangular duals.
The construction of the actual cartogram, given the area-universal
rectilinear dual and the weight function, can be accomplished
using a result by Wimer \textit{et al.} \cite{WKC}, which in turn requires numerical iteration.

The result of Eppstein {\em et al.} above is restricted to planar
graphs that have rectangular duals. Going back to the more general
rectilinear duals, leads to a series of papers where the
main goal has been to reduce the polygonal complexity while respecting
all areas and adjacencies.
 De Berg~{\em et al.}
initially showed that 40 sides suffice~\cite{deBerg07}. This was later
improved to $34$ sides~\cite{Nagamochi}. In a recent paper~\cite{BR-WADS11}
the polygonal
complexity was reduced to 12 sides and even more recently to 10
sides~\cite{ABFGKK11}.

Side contact representations of planar graphs have also been
studied without the restriction to rectilinear polygons. In the
unweighted case 6-sided polygons are sometimes necessary
and always sufficient~\cite{ghkk10}. The constructive upper bound relies on convex 6-sided polygons.
In the weighted version,
where the area of each polygon is prescribed, examples are known for which
polygons with 7~sides are necessary~\cite{ueck-phd}.
This lower bound is matched by constructive upper bound of 7~sides if holes are allowed~\cite{AlaBieFelKauKob11}.
In the same paper it is shown that even allowing arbitrarily high polygonal complexity and holes of
arbitrary size, there exist examples with prescribed areas which cannot be represented with convex polygons. If holes
are not allowed then the best previously known polygonal complexity is~10, and it is achieved with rectilinear polygons~\cite{ABFGKK11}.

\subsection{Our Results}

Recall that the known lower bound on the polygonal complexity even for
unweighted rectilinear duals is 8
while the best known upper
bound is 10. Here we present
the first construction that matches the lower bound.
Specifically, our construction produces 8-sided area-universal rectilinear duals
 in linear time, and is thus optimal in terms of polygonal
 complexity.  The exact cartogram can be computed from the area-universal rectangular layout with numerical iteration, or can be approximated with a hill-climbing heuristic.

For Hamiltonian maximal planar graphs we have an alternative construction
which allows us to directly compute cartograms with 8-sided rectilinear
polygons in linear time. Moreover, we prove that 8-sided rectilinear polygons
are necessary by constructing a non-trivial lower bound example.  If the
Hamiltonian path has the extra property that it is one-legged, then we can
reduce the polygonal complexity and realize cartograms with 6-sided polygons.
This can be used to obtain 6-sided cartograms of maximal outer-planar graphs.
Thus we have optimal (in terms of both polygonal complexity and running time)
representations for Hamiltonian maximal planar and maximal outer-planar
graphs.

\section{Preliminaries}
\label{sec:prelim}

A {\em planar graph}, $G=(V,E)$, is one that has a drawing without
crossing in the plane along with an embedding, defined via the cyclic ordering
of edges around each vertex. A {\em plane graph} is a fixed planar embedding of a planar graph.
It splits the plane into connected regions called {\em faces}; the unbounded region
is the {\em outer-face} and all other faces are called {\em interior faces}. A
planar (plane) graph is \textit{maximal} if no edge can be added to it without violating planarity.
Thus each face of a maximal plane graph is a triangle. A \textit{Hamiltonian cycle}
in a graph $G$ is a simple cycle containing all the vertices of $G$. A graph $G$ is
called \textit{Hamiltonian} if it contains a Hamiltonian cycle.

A set $P$ of closed simple
 interior-disjoint polygons with an isomorphism
 $\mathcal{P}:V\rightarrow P$ is a {\em polygonal contact
representation} of a graph
if for any two vertices
 $u, v \in V$ the boundaries of $\mathcal{P}(u)$ and $\mathcal{P}(v)$ share a non-empty
 line-segment if and only if $(u,v)$ is an edge.
Such a representation is known
 as a {\em rectilinear dual} of the input graph if polygons in $P$ are
 rectilinear.
 In the weighted version the input is the graph $G$, along with a weight function $w:V(G)\rightarrow \mathbb{R}^+$ that
 assigns a weight to each vertex of $G$. A rectilinear dual is called a \textit{cartogram} if the area
 of each polygon is equal to the pre-specified weight of the corresponding vertex.
 We define the \textit{complexity of a polygon} as the number of sides it has.
A common objective is to realize a given graph and a set of weights,
using polygons with minimal complexity.

\subsection{Canonical Orders and Schnyder Realizers}

Next we briefly summarize the concepts of a ``canonical order'' of a
planar graph~\cite{fpp-hdpgg-90} and that of a  ``Schnyder  realizer''~\cite{s-epgg-90}.
Let $G=(V, E)$ be a maximal plane graph with outer vertices $u$, $v$, $w$ in clockwise order. Then
 we can compute in linear time~\cite{cp-ltadp-95} a \textit{canonical order} or \textit{shelling order} of
 the vertices $v_1 = u$, $v_2 = v$, $v_3$, $\ldots$, $v_n = w$, which is defined as one that meets the
 following criteria for every $4\le i\le n$.

\begin{itemize}
    \item The subgraph $G_{i-1}\subseteq G$ induced by $v_1$, $v_2$, $\ldots$, $v_{i-1}$ is
     biconnected, and the boundary of its outer face is a cycle $C_{i-1}$ containing the edge
     $(u, v)$.
    \item The vertex $v_i$ is in the exterior face of $G_{i-1}$, and its neighbors in $G_{i-1}$
    form an (at least 2-element) subinterval of the path $C_{i-1}-(u, v)$.
\end{itemize}

A \textit{Schnyder realizer} of a maximal plane graph $G$ is a partition of the interior
 edges of $G$ into three sets $\mathcal{S}_1$, $\mathcal{S}_2$ and $\mathcal{S}_3$ of
 directed edges such that for each interior vertex $v$, the following conditions hold:

\begin{itemize}
    \item $v$ has out-degree exactly one in each of $\mathcal{S}_1$, $\mathcal{S}_2$ and $\mathcal{S}_3$,
    \item the counterclockwise order of the edges incident to $v$ is: entering $\mathcal{S}_1$, leaving $\mathcal{S}_2$,
    entering $\mathcal{S}_3$, leaving $\mathcal{S}_1$, entering $\mathcal{S}_2$, leaving $\mathcal{S}_3$.
\end{itemize}

Schnyder proved that any maximal plane graph has a Schnyder realizer
and it can be computed in $O(n)$ time~\cite{s-epgg-90}. The first
condition implies that $\mathcal{S}_i$, for $i=1, 2, 3$ defines a tree rooted at exactly
 one exterior vertex and containing all the interior vertices such that the edges are directed towards the root.
 Denote by $\Phi_k(v)$ the parent of vertex $v$ in tree $T_k$.
The following well-known lemma shows a profound connection between canonical orders and Schnyder realizers.

\begin{lemma}
\label{lemma:can-schny} Let $G$ be a maximal plane graph. Then the following hold.
    \begin{enumerate}
        \item[(a)] A canonical order of the vertices of $G$ defines a Schnyder realizer of $G$, where the
            outgoing edges of a vertex $v$ are to its first and last predecessor (where ``first'' is
            w.r.t. the clockwise order around $v$), and to its highest-numbered successor.
        \item[(b)] A Schnyder realizer with trees $S_1$, $S_2$, $S_3$ defines a canonical order,
            which is a topological order of the acyclic graph $\mathcal{S}_1^{-1}\cup \mathcal{S}_2^{-1} \cup \mathcal{S}_3$,
            where $\mathcal{S}_k^{-1}$ is the tree $\mathcal{S}_k$ with the direction of all its edges reversed.
    \end{enumerate}
\end{lemma}

\section{Cartograms with 8-Sided Polygons}
\label{sec:eight-side}

In this section we show that 8-sided polygons are always sufficient and sometimes necessary for a
 cartogram of a maximal planar graph.
Our algorithm for constructing 8-sided area-universal rectilinear duals has
three
main phases. In the first phase we create a contact
 representation of the graph $G$, where each vertex of $G$ is represented by
an upside-down \textbf{T}, i.e., a horizontal segment
 and a vertical segment.
 Figures~\ref{fig:8-side-illus}(a)-(b) show a maximal planar graph and
 its contact representation using \textbf{T}'s, where the three ends of each \textbf{T} are marked
with arrows. In the second phase we make both the
horizontal and vertical segments of each \textbf{T} into thin polygons
with $\lambda$ thickness for some $\lambda >0$. We then have a contact
 representation of $G$ with $T$-shaped polygons as illustrated in
 Figure~\ref{fig:8-side-illus}(c). In the third phase
 we remove all the unused area in the representation by assigning each
 (rectangular) hole to one of the polygons adjacent to it, as
 illustrated in Figure~\ref{fig:8-side-illus}(d). We show that the
 resulting representation is an area-universal rectilinear dual of $G$
 with polygonal complexity 8, as illustrated in Figure~\ref{fig:8-side-illus}(e).

\begin{figure}[htbp]
\centering
\includegraphics[width=0.9\textwidth]{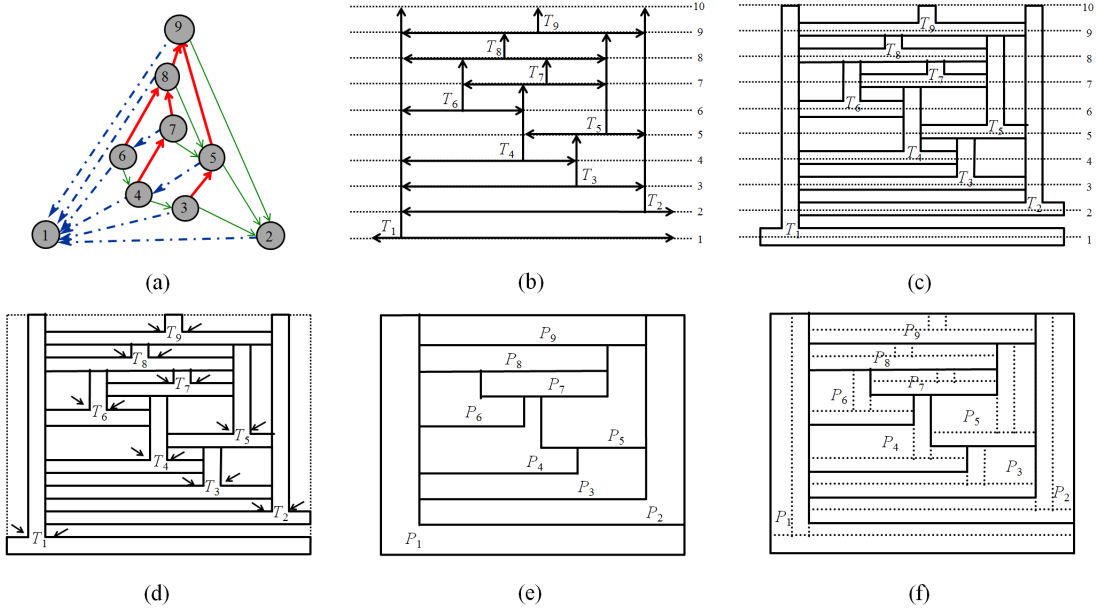}
\caption{\small\sf Construction of a rectilinear dual of a maximal planar graph with 8-sided polygons.}
\label{fig:8-side-illus}
\end{figure}

\subsection{Constructing Contact Representation with \textbf{T}'s}

Our contact representation with \textbf{T}'s is similar to the
approach described by de Fraysseix \textit{et al.}~\cite{FMR04}.

Let $G$ be a planar graph.  As mentioned earlier, we may assume
that $G$ is internally triangulated and has a simple outer-face.
If need be, we can add two vertices (which we later choose as $v_1$ and
$v_2$) and connect them to the outer-face to ensure that the graph is maximal. Now let $v_1$, $v_2$, $v_3$, $\ldots$, $v_n$ be a canonical order of the
 vertices in $G$ with corresponding Schnyder trees $\mathcal{S}_1$,
 $\mathcal{S}_2$ and $\mathcal{S}_3$ rooted at $v_1$, $v_2$ and $v_n$,
 Add to $\mathcal{S}_1$ the
 two edges $(v_2, v_1)$ and $(v_n, v_1)$ oriented towards $v_1$ and add to $\mathcal{S}_2$ the edge $(v_n, v_2)$
 oriented towards $v_2$.
In what follows, we sometimes identify vertex $v_i$ with its canonical
label $i$.

We assign to vertex $i$ the \textbf{T}-shape $T_i$
consisting horizontal and vertical segments $h_i$ and $b_i$.
Begin by placing $T_1$ and $T_2$ so that $h_1$ is placed at $y=1$, $h_2$ is placed at $y=2$, the
 topmost points of both $b_1$ and $b_2$ have $y$-coordinate $n+1$ and the leftmost point of the $h_2$ touches $b_1$.
 Next the algorithm iteratively constructs the contact representation
 by defining $T_k$ so that $h_k$ is placed at $y=k$ and the topmost point of $b_k$ has
 $y$-coordinate $\Phi_3(k)$ for $3\le k<n$.  After the $k$-th step of
 the algorithm we have a contact representation of $G_k$, and we
 maintain the invariant that the order of the
 vertical segments with non-empty parts in the half-plane $y>k$ corresponds to the same circular order of the
 vertices along $C_k-(v_1, v_2)$.

Consider inserting $T_k$ for $v_k$. The neighbors $v_{k_1}$, $v_{k_2}$, $\ldots$, $v_{k_d}$
 of $v_k$ in $G_{k-1}$ form a subinterval of $C_{k-1}-(v_1, v_2)$ and hence the corresponding vertical segments are
 also in the same order in the half-plane $y>k-1$ of the representation of $G_{k-1}$. Since $v_k=\Phi_3(v_{k_i})$ for $1<i<d$ (Lemma~\ref{lemma:can-schny}),
the topmost points of the
 corresponding vertical segments have $y$-coordinate $k$. As $v_{k_1}$ and  $v_{k_d}$ are the parents of $v_k$ in
 $\mathcal{S}_1$ and $\mathcal{S}_2$, the
 $x$-coordinates of $b_{k_1}$ and $b_{k_d}$ define the $x$-coordinates
 of the two endpoints of $h_k$. Let
%us assume that
these coordinates be
 $x_l$ and $x_r$; then $h_k$ is placed between the two points $(x_l, k)$, $(x_r, k)$ and $b_k$ is placed between
 the two points $(x_m, k)$, $(x_m, \Phi_3(k))$ with $x_l+1<x_m<x_r-1$. Finally for $k=n$, we place $T_n$ so
 that $h_n$ touches $b_1$ to the left and $b_2$ to the right and the topmost point of $b_n$ has $y$-coordinate $n+1$.

We note here that this representation can be computed in linear time so that all coordinates are integers by pre-computing a topological order
$\pi$ of ${\cal S}_1^{-1}\cup {\cal S}_2$; then  $h_k$ is the
segment $[\pi(\Phi_1(k)),\pi(\Phi_2(k))]\times k$ and $b_k$ is the
segment $\pi(k)\times [k,\Phi_3(k)]$.

\subsection{$\lambda$-Fattening of $T_i$'s}

Let $\Gamma'$ be the contact representation of $G$ using \textbf{T}'s obtained above. In this
 phase of the algorithm, we ``fatten'' \textbf{T}'s
 so that each vertex is represented by a $T$-shaped polygon. We
 replace each horizontal segment $h_i$ by an axis-aligned rectangle
 $H_i$ which has the same width as $h_i$, and whose top
 (bottom) side is $\lambda /2$ above (below) $h_i$, for some $0<\lambda$,
 as illustrated in Figure~\ref{fig:subdivision}(a).
 Similarly, we replace each vertical segment $b_i$ by an
 axis-aligned rectangle $B_i$ which has the same height as $b_i$ and whose left (right) side is $\lambda /2$ to the left (right) of $b_i$.
 We call this process \textit{$\lambda$-fattening} of $T_i$. Note that
 this process creates intersections of $H_i$ with $B_i$,
 $B_{\Phi_1(i)}$ and $B_{\Phi_2(i)}$ and intersection of $B_i$ with $H_{\Phi_3(i)}$. We remove these intersections
 by replacing $H_i$ by $H_i-B_{\Phi_1(i)}-B_{\Phi_2(i)}$ and replacing $B_i$ by $B_i-H_i- H_{\Phi_3(i)}$. The
 resulting layout is a contact representation $\Gamma ''$ of $G$ where each vertex $v_i$ of $G$ is
 represented by the $T$-shaped polygon $H_i\cup B_i$.

\begin{figure}[htbp]
\centering
\includegraphics[width=0.7\textwidth]{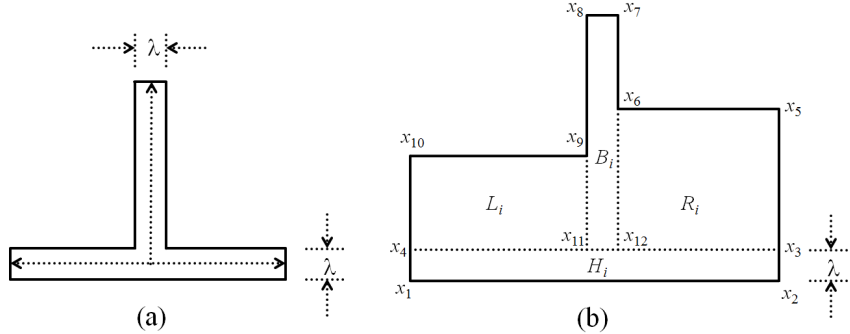}
\caption{\small\sf (a) $\lambda$-fattening of $T$, and (b) subdividing a $T$-shaped polygon into four rectangles.}
\label{fig:subdivision}
\end{figure}

\subsection{Removing unused area}

In this step, we begin with the $\lambda$-fat $T$-shaped polygonal
layout, $\Gamma ''$, from above and
assign each
(rectangular) hole to a polygon adjacent to it. We start by placing an axis-aligned
 rectangle of minimum size that encloses $\Gamma ''$. This creates five new bounded holes.
Note that all the holes in $\Gamma ''$ are rectangles, and
each of them is bounded at the bottom by $H_i$ for some vertex $v_i$.
We assign each hole to this vertex.  This assigns at most two holes
to each vertex $v_i$: one hole $L_i$ to the left of $B_i$, and one
hole $R_i$ to the right of $B_i$.
Now for each
vertex $v_i$, define $P_i=T_i\cup L_i\cup R_i$. It is easy to see
 that $P_i$ is an 8-sided rectilinear polygon since the left side of $L_i$ has the same $x$-coordinate as the left side
 of $H_i$ and the right side of $R_i$ has the same $x$-coordinate as
 the right side of $H_i$. Thus we have a rectilinear dual, $\Gamma$,
 of $G$ where each vertex $v_i$ is represented by $P_i$.

We preferred the above description for the computation of $\Gamma$ since it gives the reader some intuition for
 the construction. However, we note here that the coordinates of $P_i$ could be computed directly, without going
 through T-shapes and $\lambda$-fattening, using the values $\Phi_k(v_i)$ for $k=1,2,3$ and a topological order
 $\pi$ of ${\cal S}_1^{-1}\cup {\cal S}_2$. To this end, we take a topological ordering of the acyclic
 graph $\mathcal{S}_1^{-1}\cup\mathcal{S}_2$ and for each vertex $v_i\neq v_n$, we denote the index of $v_i$ in
 this topological ordering by $\pi(i)$. We can fix the placement of the left and right side of $B_i$ at
 $x$-coordinate $\pi(i)-\lambda/2$ and $\pi(i)+\lambda/2$, respectively. Then for each vertex $v_i$,
 the 8-sided rectilinear polygon $P_i$ representing $v_i$ is defined as follows.

\begin{itemize}
	\item The horizontal \textit{base segment} of $P_i$ has $y$-coordinate $i-\lambda/2$ and extends from the left
		to the right segment (defined below.)
	\item The horizontal \textit{top segment} of $P_i$ has $y$-coordinate $\Phi_3(i)$ (in case of $P_n$, the top segment
		has $y$-coordinate $n+\lambda/2$).
	\item The vertical \textit{left segment} of $P_i$ has $x$-coordinate $\pi(\Phi_1(i)) + \lambda/2$ and goes upward
		from the base segment (in case of $P_1$ and $P_n$, the left segment has $x$-coordinate $1-\lambda/2$).
	\item The vertical \textit{right segment} of $P_i$ has $x$-coordinate $\pi(\Phi_2(i)) - \lambda/2$ and goes upward
		from the base segment (in case of $P_2$ and $P_n$, the right segment has $x$-coordinate $n-1+\lambda/2$).
	\item If $v_i$ has no children in $\mathcal{S}_2$, then the left segment extends upward until the top segment. Otherwise
		let $j$ be the vertex in it that comes clockwise after $\Phi_1(i)$ in the order of neighbors around $i$. (One can
		see that $v_j$ is the child of $v_i$ in $\mathcal{S}_2$ with the lowest canonical number.) In this case, the left
		segment extends upward until $x$-coordinate $j-\lambda/2$, from which point the polygon continues rightward towards
		the left reflex corner at $\pi(i) - \lambda/2$, and then upward until it meets the top segment.
	\item If $v_i$ has no children in $\mathcal{S}_1$, then the right segment extends upward until the top segment. Otherwise
		let $j$ be the vertex in it that comes counter-clockwise after $\Phi_2(i)$ in the order of neighbors around $i$.
		(One can see that $v_j$ is the child of $v_i$ in $\mathcal{S}_1$ with the lowest canonical number.) In this case,
		the right segment extends upward until $x$-coordinate $j-\lambda/2$, from which point the polygon continues leftward
		towards the right reflex corner at $\pi(i) + \lambda/2$, and then upward until it meets the top segment.
\end{itemize}

Then the union of these $n$ polygons define the rectilinear dual $\Gamma$ of $G$ which is contained inside the rectangle
 $[1-\lambda/2, n-1+\lambda/2]\times[1-\lambda/2, n+\lambda/2]$. Thus we can compute the representation in linear time,
 and by scaling the representation by a constat factor, we can make all coordinates to be integers of size $O(n)$.

\subsection{Area-Universality}
\label{sec:universal}

A rectilinear dual $\Gamma$ is \textit{area-universal} if any assignment of areas to its polygons can be realized by a
 combinatorially equivalent layout. Eppstein \textit{et al.}~\cite{EMVS} studied
this concept
 for the case when all the polygons are rectangles and the outer-face
 boundary is also a rectangle (which they call a {\em rectangular layout}). They gave a characterization of area-universal rectangular layouts using the concept of
 ``maximal line-segment''.  A \textit{line-segment} of a layout is the union of inner edges of the layout forming a
 consecutive part of a straight-line. A line-segment that is not contained in any other line-segment is maximal. A
 maximal line-segment $s$ is called {\em one-sided} if it is part of the side of at least one rectangular face, or
 in other words, if the perpendicular line segments that attach to its interior are all on one side of $s$.

\begin{lemma}~\cite{EMVS}
\label{lem:area-uni} A rectangular layout
is area-universal if and only if each maximal segment in the layout is one-sided.
\end{lemma}

No such characterization is known when some faces are not rectangles.
Still we can use the characterization from
Lemma~\ref{lem:area-uni} to show that the rectilinear dual obtained by the algorithm from
the previous section is area-universal, with the following Lemma.

\begin{lemma}
\label{lem:eight-suff} Let $\Gamma$ be the rectilinear dual
obtained by the above algorithm. Then $\Gamma$ is area-universal.
\end{lemma}
\begin{proof} To show the area-universality of $\Gamma$, we divide all the
polygons in $\Gamma$ into a set of rectangles such that the resulting
rectangular layout is area-universal. Specifically, we divide each polygon $P_i$
into four rectangles $H_i$ $B_i$, $L_i$ and $R_i$ (as defined in the previous subsection)
by adding three auxiliary segments: one horizontal and two vertical, as illustrated in
Figure~\ref{fig:subdivision}(b).
Any horizontal segment $s$ not on the bounding box belongs to some $H_i$
(either top or bottom), and expanding it to its maximum it ends at
$B_{\Phi_1(v_i)}$ on the left and $B_{\Phi_2(i)}$ on the right.  So $s$
is one-sided since it is a side of the rectangle $H_i$.
Any vertical segment $s$ not on the bounding box belongs to some $B_i$
(either left or right), and expanding it to its maximum it ends at
$H_i$ on the bottom and $H_{\Phi_3(i)}$ on the top.  So $s$
is one-sided since it is a side of the rectangle $B_i$.

Now given any assignment of areas $w:V\rightarrow \mathbb{R}^+$ to the vertices $V$ of $G$, we split $w(v_i)$ arbitrarily
 into four parts and assign the four values to its four associated rectangles. Since $\Gamma ^*$ is area-universal,
 there exists a rectilinear dual of $G$ that is combinatorially equivalent to $\Gamma$ for which these areas are realized.
Figure~\ref{fig:8-side-illus}(f) illustrates the rectangular
layout obtained from the rectilinear dual in Figure~\ref{fig:8-side-illus}(e).
\end{proof}

So for any area-assignment, the rectilinear dual that we found can
be turned into a combinatorially equivalent one that respects the
area requirements.  This proves our main result for maximal planar
graphs. Omitting $v_1$ and $v_2$ from the drawing
still results in a cartogram where the union of all polygons is a
rectangle, so the result also holds for all planar graphs that are
inner triangulated and have a simple outer-face.

Recall that the lower bound on the complexity
of polygons in any rectilinear dual (and hence in any cartogram) is 8,
as proven by Yeap and Sarrafzadeh~\cite{ys93}. The algorithm described
in this section, along with this lower bound leads to our main theorem.

\begin{theorem}
\label{th:eight-nece-suf} Eight-sided polygons are always sufficient and sometimes necessary for
 a cartogram of an inner triangulated planar graph with a simple outer-face.
\end{theorem}

\subsection{Feature Size and Supporting Line Set}

In addition to optimal polygonal complexity, we point out here that the 8-sided area-universal
 rectilinear layout constructed with our algorithm maximizes the feature size and reduces the
 number of supporting lines. Earlier constructions, e.g.,~\cite{deBerg07,BR-WADS11},
often rely on ``thin connectors'' to maintain adjacencies, whereas our
construction does not. Specifically, let $G$
be a maximal planar graph with a prescribed weight function $w:V(G)\rightarrow
\mathbb{R}^+$. Choose $W$ and $H$ such that $W\times H=A=\sum_{v\in V(G)}w(v)$. We are
interested in cartograms within a rectangle of width $W$ and the height $H$.
Define $w_{min}=min_{v\in V(G)}w(v)$.

Recall that each vertex $v_i$ is represented by the union of at most
four rectangles $H_i\cup B_i\cup R_i \cup L_i$, with $H_i$ and $B_i$
non-empty.  We can distribute the weight of $v_i$ arbitrarily among them.
In particular, we can
assign zero areas to the rectangles $L_i$ and $R_i$ and split its weight into
two equal parts to $H_i$ and $B_i$.
In this layout each original vertex is represented by rectangles
$H_i$ and $B_i$ whose union is either a rectangle or some fattened $T$ or $L$,
and all the necessary
contacts remain. Hence we can use this simplified layout to produce the
cartogram.

The distribution of the weight of $v_i$ in equal parts to $H_i$ and
$B_i$ allows to bound the feature size.  The height and width of each
rectangle are bounded by $H$ and $W$ respectively.  Its weight is at
least $w_{min}/2$. Therefore, the height and width of each rectangle is at
least $\frac{w_{min}}{2 \max\{W, H\}}$. Thus the minimum feature size of
the cartogram is at least $\frac{w_{min}}{2 \max\{W, H\}}$. This is
worst-case optimal, as the polygon with the smallest weight might need
to reach from left to right and top to bottom in the
representation. We may choose $W=H=\sqrt{A}$. Then the minimum feature
size is $\frac{w_{min}}{2\sqrt{A}}$. Furthermore the rectangular
layout based on the rectangles $H_i$ and $B_i$ alone yields a
cartogram with at most $2n$ supporting lines, instead of the $3n$
supporting lines in the cartogram based on four rectangles per vertex.

\subsection{Computing the Cartogram}

The proof of Lemma~\ref{lem:area-uni} implies an algorithm for computing the final cartogram. Splitting the $T$-shaped polygons
 into four rectangles and distributing the weights on these rectangles yields an
area-universal rectangular dual.  This combinatorial structure has to be
turned into an actual cartogram, i.e., into a layout respecting the given
weights.
Wimer \textit{et al.}~\cite{WKC} gave a formulation of the problem
which combines flows and quadratic equations.  Eppstein \textit{et al.}~\cite{EMVS} indicated that a solution can be
 found with a numerical iteration.  Alternate methods also exist, based on non-linear programming~\cite{rosenberg},
 geometric programming~\cite{MCH96}, and convex programming~\cite{ChenFan98}. Heuristic hill-climbing schemes converge
 much quicker and can be used in practice, at the expense of small errors~\cite{Ceder92,ITK98,WC95}.

\subsection{Implementation and Experimental Results}

We implemented the entire algorithm, along with a force-directed heuristic to compute the final cartogram.
We treat each region as a rectilinear ``room'' containing an amount of ``air'' equal to the weight assigned to the corresponding vertex. We then simulate the natural phenomenon
 of air pressure applied to the ``walls'', which correspond to the line segment borders in our layout.

\begin{figure}[htbp]
\centering
\includegraphics[width=0.45\textwidth]{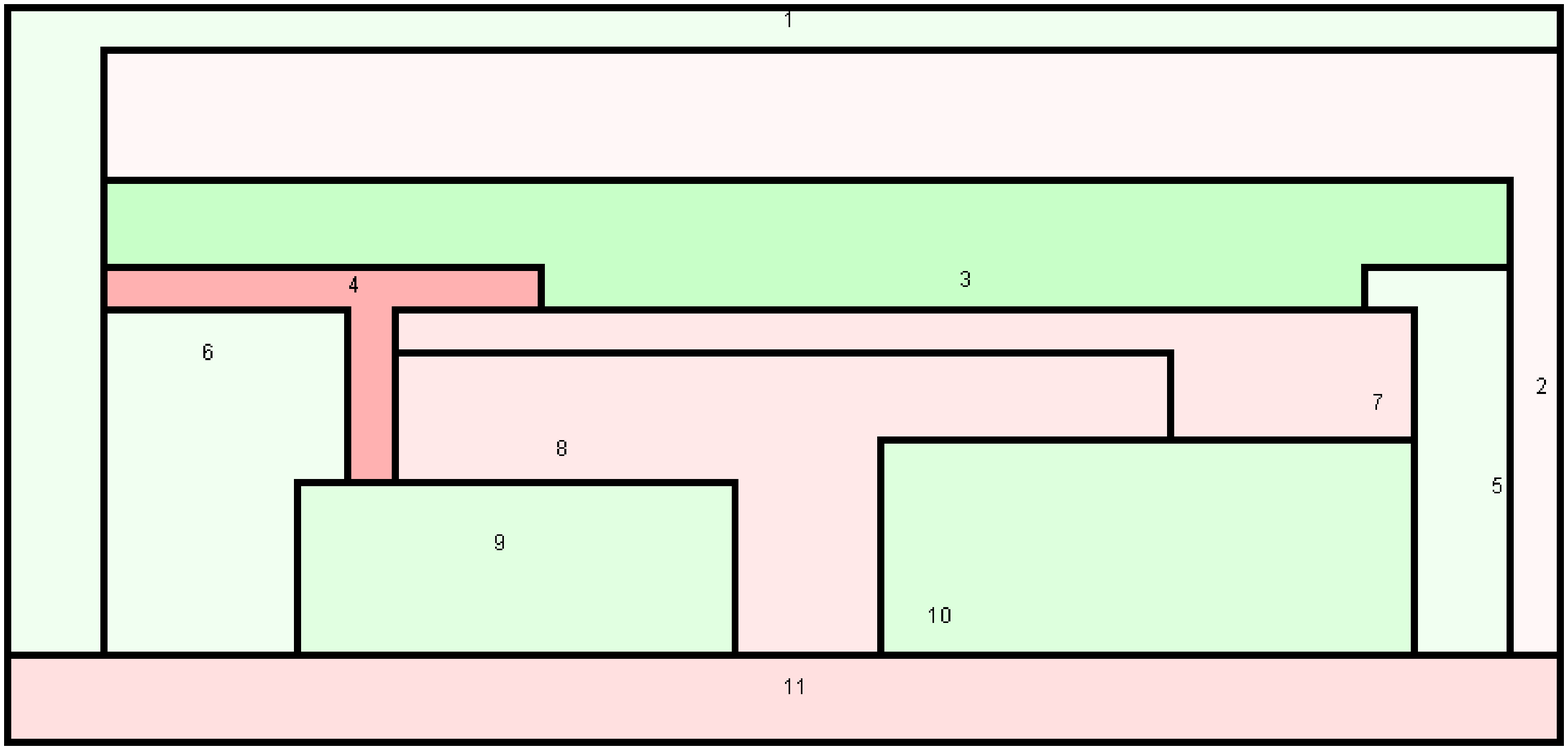}\hspace{.4cm}
\includegraphics[width=0.45\textwidth]{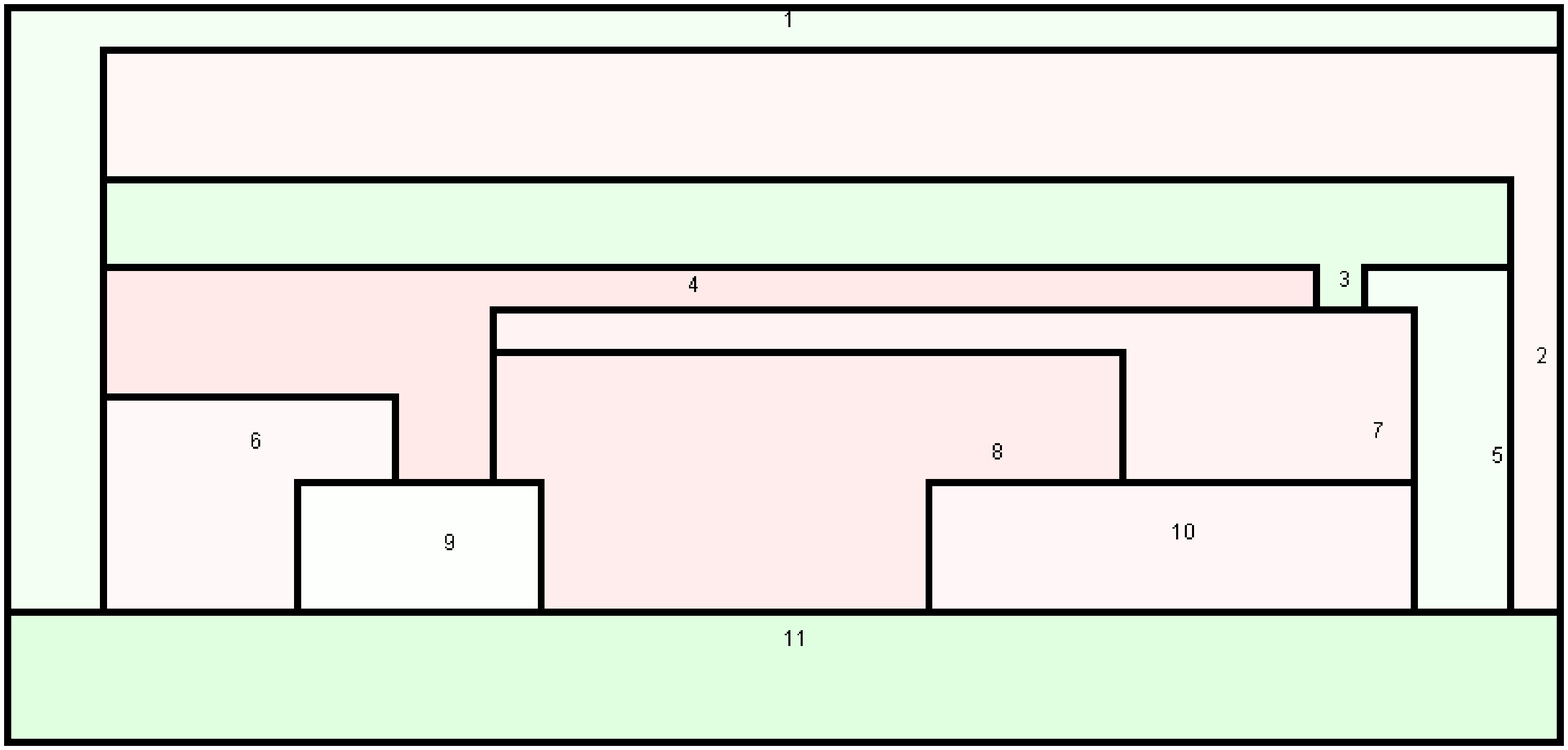}\\
\vspace{.5cm}
\includegraphics[width=0.45\textwidth]{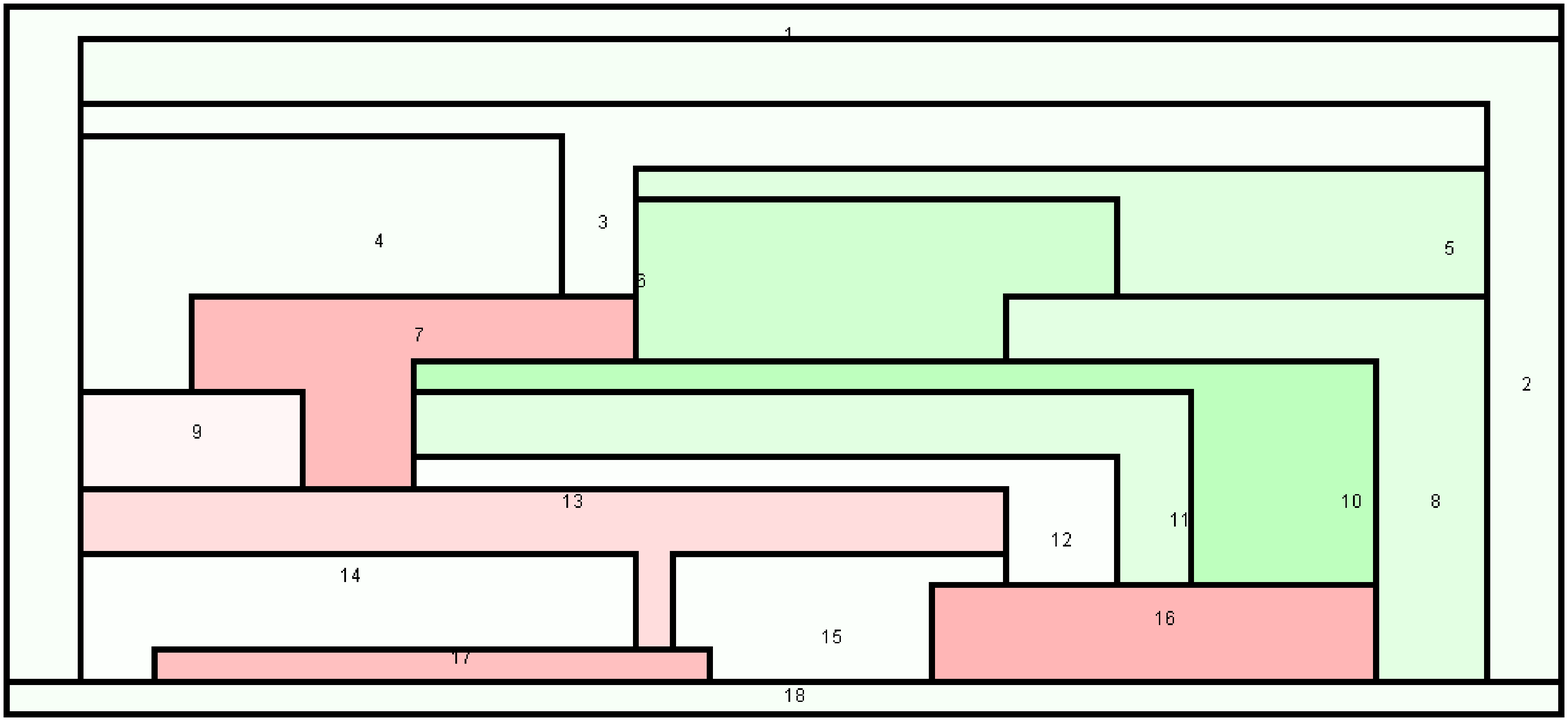}\hspace{.4cm}
\includegraphics[width=0.45\textwidth]{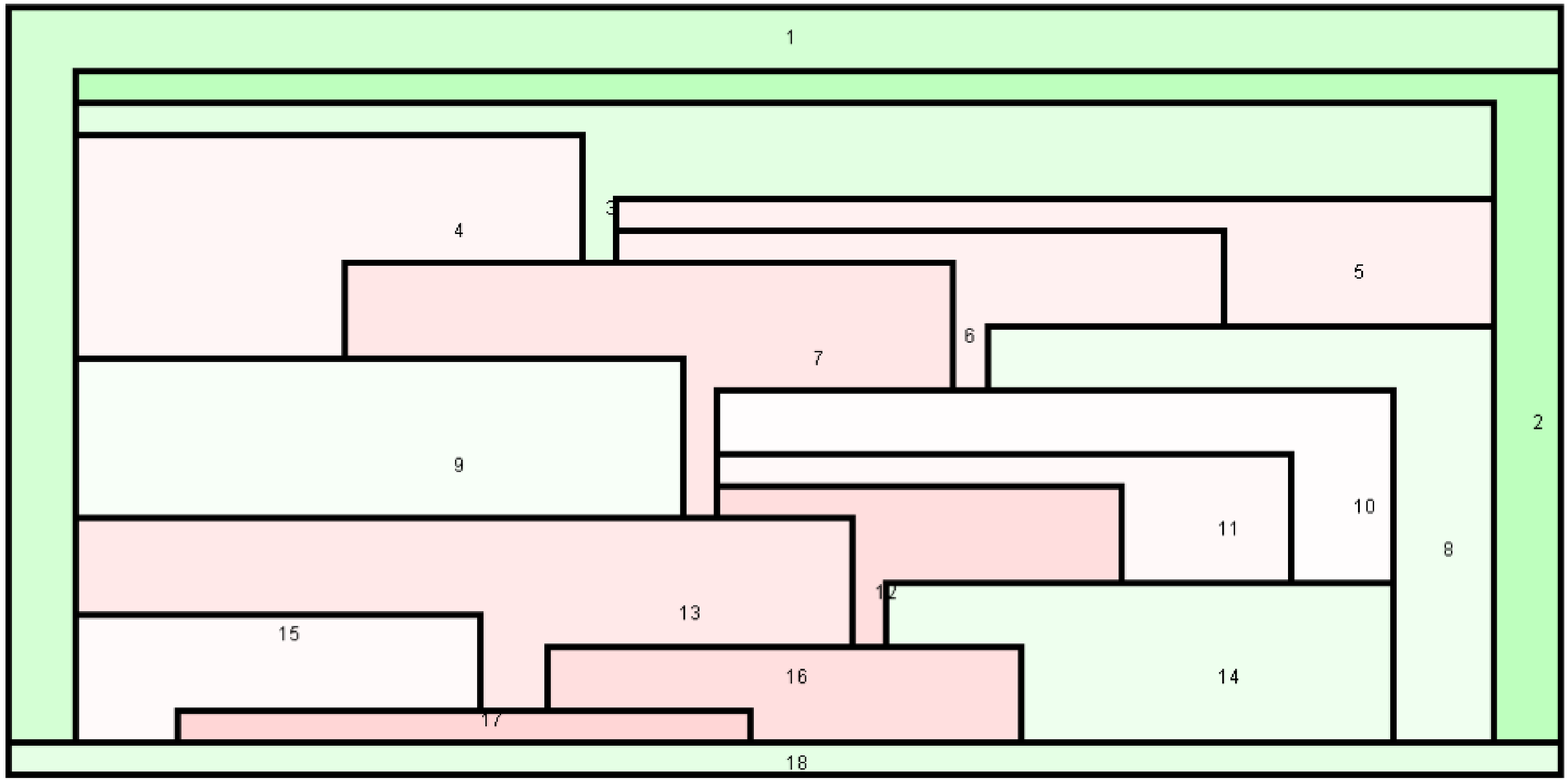}
\caption{\small\sf Input-output pairs: weights are assigned at random in the range $[10,100]$, and the cartographic error in the output file is less than 1\%. The colors indicate air-pressure: the greener a region is, the more it needs to shrink; the redder a regions is, the more it needs to grow.}
\label{fig:examples}
\end{figure}

For each vertex $v_i$ of $G$, the polygon $P_i$ contains air with volume $w(v_i)$. If the area
 of $P_i$ is $A_i$, then the pressure applied to each of the walls surrounding $P_i$ is given by
 $\mathcal{P}(v_i)=\frac{w(v_i)}{A_i}$. In Section~\ref{sec:universal}, we saw that the maximal segments of the layout
 are the two horizontal and the two vertical segment associated with each polygon. For each polygon,
 the horizontal segment other than the base is entirely inside the polygon, hence it feels no
 ``pressure'' on it. For each of the other three segments $s$ for the polygon $P_i$, the ``inward force''
 it feels is given by $\mathcal{F}(s)=\sum_{v_j\in V(s)}[\mathcal{P}(v_j)l_j]-\mathcal{P}(v_i)l_i$.
 Here $V(s)$ is the set of vertices other that $v_i$ whose corresponding polygon touches the segment $s$
 and $l_i$ (resp. $l_j$) denotes the length of $s$ that is shared with $P_i$ (resp. $P_j$).
At each iteration,
 we consider the segment that feels the maximum pressure and let it move in the appropriate direction. The convergence of this scheme
 follows from~\cite{ITK98}. Some sample input-output pairs are shown in Fig.~\ref{fig:examples}; more examples and movies showing the gradual transformation can be found at \url{www.cs.arizona.edu/~mjalam/optocart}.

We ran a few simple experiments to test the heuristic for time and accuracy. In the first experiment we generated 5 graphs on $n$ vertices with $n$ in the range $[10-50]$ and assigned 5 random weight distributions with weights in the range $[10-100]$. Next we ran the heuristic until the cartographic error dropped below 1\% and recorded the average time. All the averages were below 50 milliseconds, which confirms that good solutions can be found very quickly in practice; see Fig~\ref{fig:exp}(a).
In the second experiment we fixed the time allowed and tested the quality of the cartograms obtained within the time limit. Specifically, we generated 5 graphs on $n$ vertices with $n$ in the range $[10-50]$
and assigned 5 random weight distributions with weights in the range $[10-100]$. We allowed the program to run for 1 millisecond and recorded the average ``cartographic error''. Even with such a small time limit, the average cartographic error was under 2.5\%; see Fig.~\ref{fig:exp}(b). Here, the \textit{cartographic error} for a cartogram of a planar graph $G=(V, E)$ is defined as in~\cite{ks07}: $max_{v\in V}(|A(v)-w(v)|/w(v))$, where $w(v)$ denotes the weight assigned to $v$ and $A(v)$ denotes the area of the polygon representing $v$. All of the experiments were run on an Intel Core i3 machine with a 2.2GHz processor and 4GB RAM.

\begin{figure}[htbp]
\centering
\includegraphics[width=0.49\textwidth]{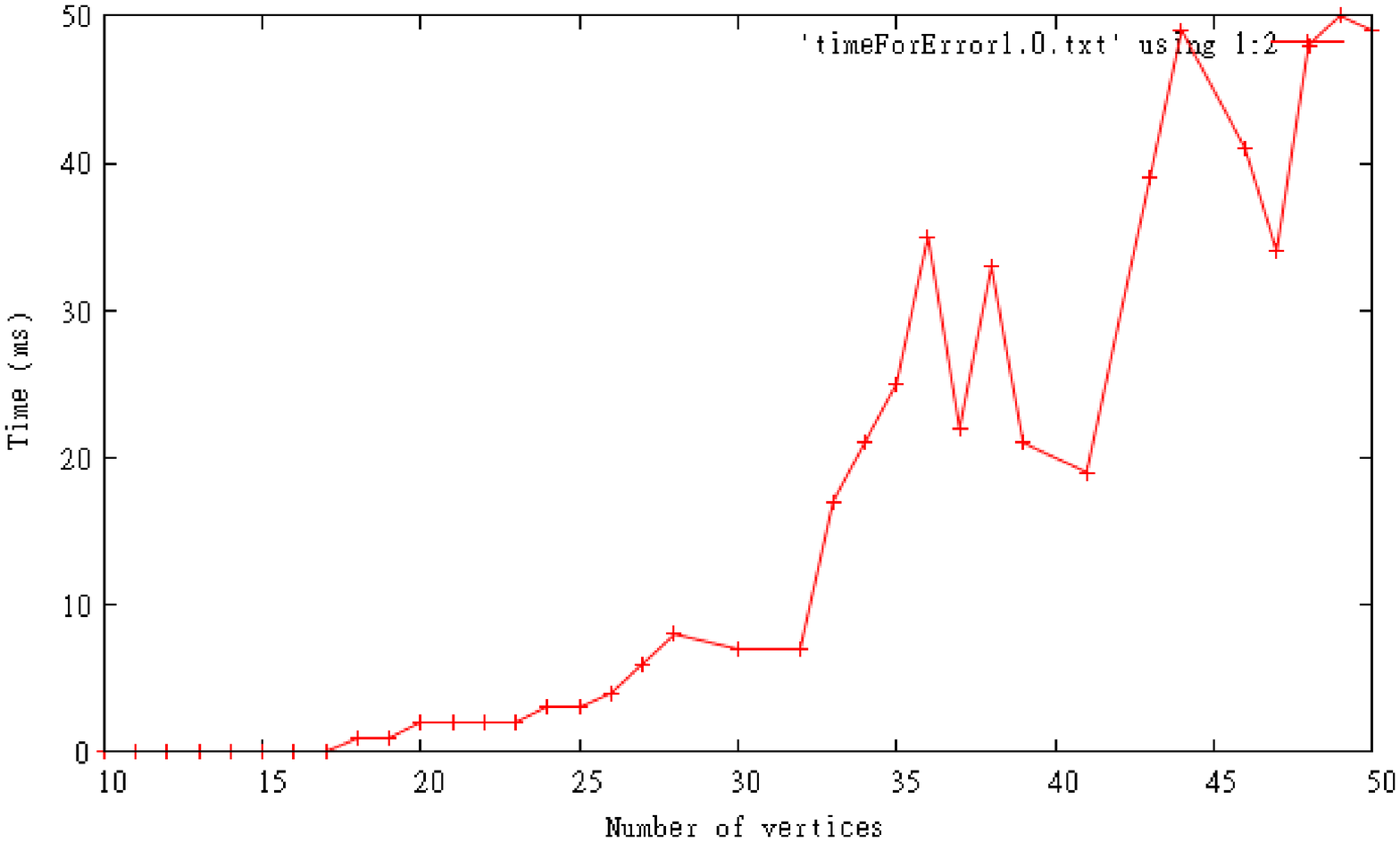}\hspace{.0cm}
\includegraphics[width=0.49\textwidth]{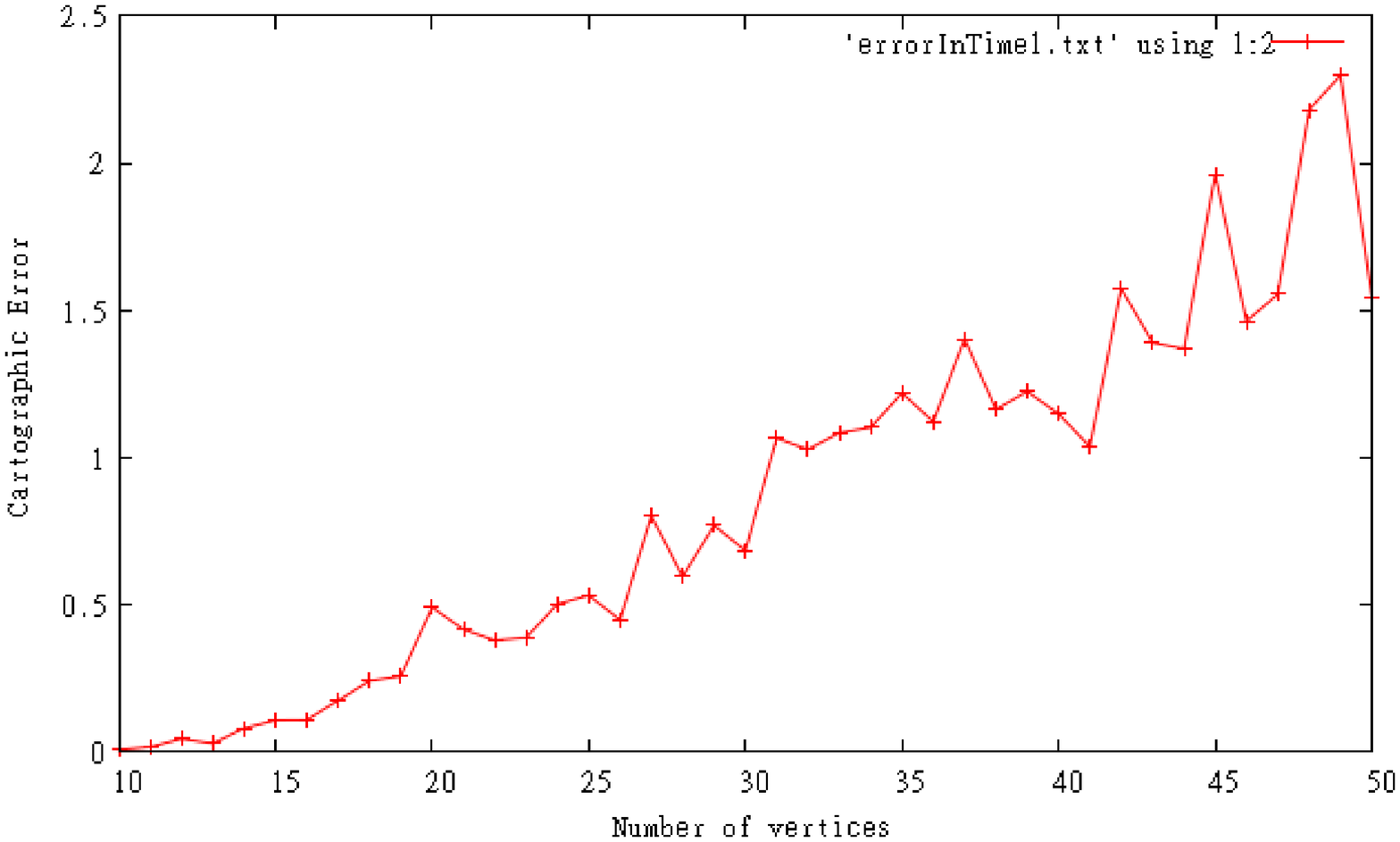}\\
 (a) \hspace{7cm} (b)\\
\caption{\small\sf Experimental results for graphs with 10-50 vertices. Each sample point corresponds to 25 graphs. (a) Plotting the average time it takes to reach cartographic error of 1\%. (b) Plotting the average cartographic error achieved in 1millisecond. }
\label{fig:exp}
\end{figure}

\section{Cartograms for Hamiltonian Graphs}
\label{sec:ham}

In this section we show that 8-sided polygons are always sufficient
and sometimes necessary for a cartogram of a {\em
 Hamiltonian} maximal planar graph.
 We first give a direct linear-time construction with 8-sided regions
 without relying on numerical iteration or heuristics, as discussed in the previous section.
 We then prove that this is optimal by showing that 8 sides are necessary, with a non-trivial lower bound example.

\subsection{Sufficiency of 8-sided Polygons}

Let $v_1,\dots,v_n$ be a Hamiltonian cycle of a maximal planar graph
  $G$. Consider a plane embedding of $G$ with the edge $(v_1,v_n)$ on the triangular outer-face. The Hamiltonian cycle splits the plane
 graph $G$ into two outer-planar graphs which we call the {\em left
   graph $G_l$} and {\em right graph $G_r$}. Edges on the Hamiltonian cycle belong to both graphs.
 The naming is with respect to a planar drawing $\Gamma$ of $G$ in which the
vertices
 $v_1,\dots,v_n$ are placed in increasing order along a vertical
 line, and the edges are drawn with $y$-monotone curves with leftmost
 edge $(v_1,v_n)$; see Figure~\ref{fig:c-shape-illus}(a).

\begin{lemma}\label{lem:ham-eight}
 Let $G=(V,E)$ be a Hamiltonian maximal planar graph and let $w:V\rightarrow \mathbb{R}^+$ be a weight
 function. Then a cartogram with 8-sided polygons can be computed in linear time.
\end{lemma}
\begin{proof}
Let $v_1, \ldots, v_n$ be a Hamiltonian cycle and $\Gamma$ be the
drawing defined above with $(v_1,v_n)$ on the outer-face.
Suppose $R$ is a rectangle of width $W$ and height $H$
where $W\times H=A=\sum_{v\in V}w(v)$.
Each vertex $v_i$ will be represented as the union of three
 rectangles, the \textit{left leg}, the \textit{body} $B_i$, and
\textit{right leg} of $v_i$. We set the width of the legs to $\lambda _i=w(v_i)/(2H+W)$;
see Figure~\ref{fig:c-shape-illus}(b).

Our algorithms places vertices $v_1,\dots,v_n$ in this order, and also
reserves vertical strips for legs of all vertices that have earlier neighbors.
More precisely, let $\calL_j$ be all vertices $v_k$ with an edge $(v_i,v_k)$
in $G_l$ for which $i\leq j<k$.  Similarly define $\calR_j$ with respect
to edges in $G_r$.  In the drawing $\Gamma$, $\calL_j$ are those vertices above
$v_j$ for which the horizontal ray left from $v_j$ crosses an incident edge.

We place vertices $v_1,\dots,v_j$ with the following invariant: The
horizontal line through the top of $B_j$ intersects, from left to right:
(a) a vertical strip of width $\lambda_k$ for
each $v_k\in \calL_j$, in descending order, (b) a non-empty part of the
top of $B_j$, and (c) a vertical strip of width $\lambda_k$ for
each $v_k\in \calR_j$, in ascending order.

We start by placing $B_1$ as a rectangle that spans the bottom of $R$.
At the left and right end of the top of $B_1$, we reserve vertical
strips of width $\lambda_k$ for each vertex in $\calL_1$ and $\calR_1$, respectively.

To place $B_i$, $i>1$, first locate the vertical strips reserved for
$v_i$ in previous steps (since $v_i\in \calL_{i-1}$ and $i\in \calR_{i-1}$,
there always are such strips, though they may have started only
at the top of $B_{i-1}$).  Since vertical strips are in
descending/ascending order, the strips for $v_i$ are the innermost ones.
Let $B_i$ be a rectangle just above $B_{i-1}$ connecting these strips.
Choose the height of $B_i$ so large that it, together with the left
and right leg inside the strips, has area $w(v_i)$; we will discuss
soon why this height is positive.

Finally, at the top left of the
polygon of $v_j$ we reserve a new vertical
strip of width $\lambda_k$ for each vertex $k$ that is in $\calL_i-\calL_{i-1}$.
Similarly reserve strips for vertices in $\calR_i-\calR_{i-1}$.   Using
planarity, it is easy to see that vertices in $\calL_i-\calL_{i-1}$ must have smaller
indices than vertices in $\calL_{i-1}$, and so this can be done such that the
order required for the invariant is respected.

\begin{figure}[htbp]
%\vspace{-0.4cm}
\centering
\includegraphics[width=0.98\textwidth]{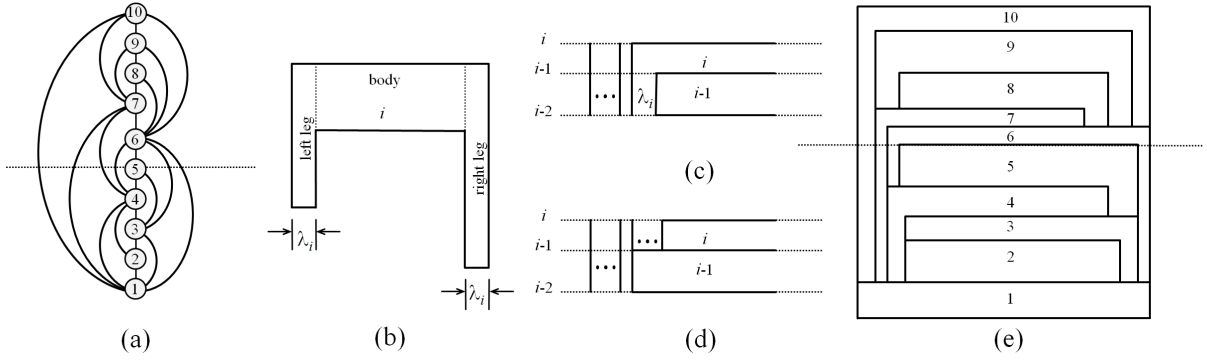}
\caption{\small\sf (a) A Hamiltonian maximal planar graph $G$, (b) an 8-sided polygon for vertex $i$, (c)--(d) illustration
 for the algorithm to construct a cartogram of $G$, (e) a cartogram of $G$ with 8-sided polygons.}
\label{fig:c-shape-illus}
\end{figure}

Clearly this algorithm takes linear time and constructs 8-gons of the
correct area.  To see that it creates contacts for all edges, consider
an edge $(v_i,v_k)$ with $i<k$ in $G_l$ (edges in $G_r$ are similar.)
By definition $k\in \calL_i$.  If $v_k\in \calL_i-\calL_{i-1}$,
then we reserved a vertical strip for $v_k$ when placing $v_i$.  This
vertical strip is used for the left leg of $v_k$, which hence touches $v_i$.
Otherwise ($v_k\not\in \calL_i-\calL_{i-1}$)
we have $v_k\in \calL_{i-1}$. At the time that $v_{i-1}$ was placed, there
hence existed a vertical strip for $v_k$.
There also was a vertical strip for $v_i\in \calL_{i-1}$.
These two strips must be adjacent, because
by planarity (and edge $(v_i,v_k)$) there can
be no vertex $v_j$ with $i<j<k$ in $\calL_{i-1}$.
So these strips create a contact between the two left legs of $v_i$ and $v_k$.

We now discuss the choice of $\lambda_i=w(v_i)/(2H+W)$.
Each leg of $v_i$ has height $\leq H$ and width $\lambda_i$, hence area
$\leq H\lambda_i$.  Then the body $B_i$ has area $\geq w(v_i)-2H\lambda_i$
and width $\leq W$, hence height $\geq (w(v_i)-2H\lambda_i)/W=\lambda_i$.
It follows that $B_i$ has positive height.  Also all vertical
strips fit:  after placing vertex $v_i$, we have a strip of width $\lambda_k$
for each vertex $v_k$ in $\calL_i$ and $\calR_i$, and  these
strips use width
$$ \sum_{v_k\in \calL_i} \frac{w(v_k)}{2H{+}W} + \sum_{v_k\in \calR_i} \frac{w(v_k)}{2H{+}W}
\leq \frac{2\sum_{v_k\in V-\{v_i\}} w(v_k)} {2H{+}W}
\leq \frac{2(A-w(v_i))} {2H{+}W}
\leq \frac{2A}{2H} - \frac{2w(v_i)}{2H{+}W} = W-2\lambda_i.$$
Hence $B_i$ has width $\geq 2\lambda_i$ and the polygon
of $v_i$ has minimum feature size $\lambda_i$.
\end{proof}

This algorithm also guarantees a minimum feature
 size for the cartogram: $\min_{v_i\in V}\lambda_i=\frac{w_{min}}{2H{+}W}$, where $w_{min}=\min_{v\in V}w(v_i)$. Choosing $W=\sqrt{2A}$ and $H=\sqrt{A/2}$, yields minimum feature size $\frac{w_{min}}{2\sqrt{2}\sqrt{A}}$.

\subsection{Necessity of 8-sided Polygons}
While it was known that 8-sided rectilinear polygons are necessary for
general planar graphs~\cite{Rin87}, here we show that 8-sided
rectilinear polygons are necessary even for Hamiltonian maximal planar
graphs.

\begin{figure}[htbp]
\hspace*{\fill}
\scalebox{1.2}{\input{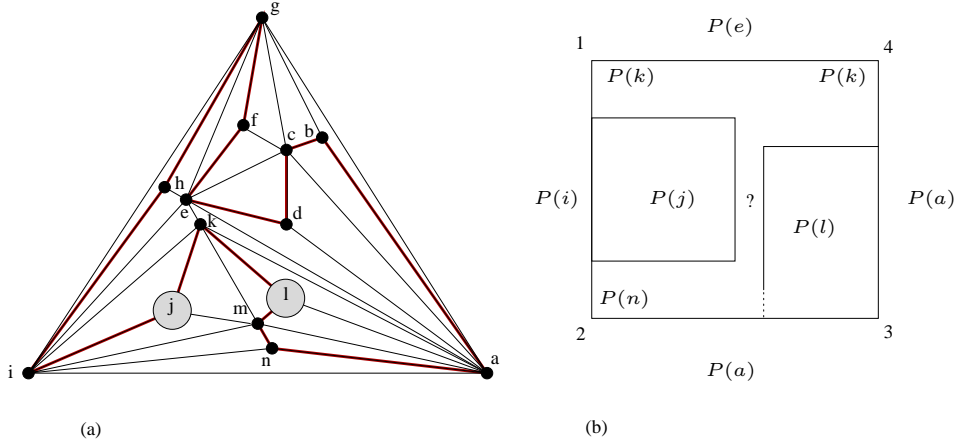}}
\hspace*{\fill}\
 \caption{\small\sf (a) A maximal planar Hamiltonian graph with a weight function that requires at least one 8-sided polygon
  in any cartogram. (b) Illustration for the proof of Lemma~\ref{lem:ham-counter}.}
 \label{fig:ham-counter}
\end{figure}

\begin{lemma}
\label{lem:ham-counter} Consider the Hamiltonian maximal planar graph $G = (V,E)$ in Figure \ref{fig:ham-counter}(a).
 Define $w(j) = w(l) = D$ and $w(v) = \delta$ for $v \in V \setminus\{j,l\}$, where $D\gg\delta$. Then any cartogram
 of $G$ with weight function $w$ requires at least one 8-sided polygon.
\end{lemma}
\begin{proof}
Assume for a contradiction that $G$ admits a cartogram $\Gamma$ with respect to $w$ such that the polygons $\{P(v)\}$
used in $\Gamma$ have complexity at most 6.
Observe that if $\{u,v,x\}$ is some
 separating triangle in $G$, i.e.,
three mutually adjacent vertices whose removal disconnect the graph,
then the region $R_{uvw}$ used for the inside of the separating triangle contains
at least one reflex corner of the polygon $P(u)$, $P(v)$, or $P(x)$.
 The $5$-vertex set $\{a,c,e,g,i\}$ in $G$ is the union of the five separating triangles
 $\{a,c,g\}$, $\{a,c,e\}$, $\{c,e,g\}$, $\{a,e,i\}$, and $\{e,g,i\}$ with disjoint interiors. Since all the
 polygons in $\Gamma$ are either $4$-sided or $6$-sided,
 the union of
 the polygons for these five vertices has
 at most five reflex corners and hence each of the five separating triangles above contains the only reflex
 corner of the polygon for $a$, $c$, $e$, $g$, or $i$.
In particular, the outer boundary of $R_{aei}$
 contains exactly one reflex corner from one of $P(a)$, $P(e)$ and $P(i)$, and hence it is a rectangle, say
 $1234$.
By symmetry, we may assume that the reflex corner of $P(i)$ is {\em not}
used for $R_{aei}$.

The $4$-vertex set $\{a,i,k,m\}$ is the disjoint union of
three separating triangles $\{a,k,m\}$, $\{k,i,m\}$, $\{i,a,m\}$
whose interiors are vertices $l,j$ and $n$, respectively.
Since the reflex corner of $P(i)$ is not used for $R_{aei}$, it
also cannot be used for any of these separating triangles. Hence each
of $P(j)$, $P(l)$ and $P(n)$ contains exactly one reflex
corner from $P(a)$, $P(k)$ and $P(m)$.  In particular,
rectangle $1234$ (which is $R_{aei}$) must contain the reflex corner
of $P(a)$.   We also can conclude that
$P(j)$, $P(l)$ and $P(n)$ are all rectangles, since there are no
additional reflex corners available to accommodate additional
convex corners from $P(j)$, $P(l)$ and $P(n)$.

Assume the naming in Figure \ref{fig:ham-counter}(b) is such
that edge $12$ belongs to $P(i)$, edges $23$ and $34$ belong to $P(a)$
and edge $41$ belongs to $P(e)$.   By the adjacencies, $P(k)$ must occupy
corners 1 and 4 and $P(n)$ must occupy corner 2, while corner 3 (which
is the reflex corner of $P(a)$) could belong to $n$ or $l$.

Now consider the rectangles $P(j)$ and $P(l)$.  If $D$ is sufficiently big,
then these two rectangles each occupy almost half of rectangle $1234$.
Therefore, either their $x$-range or their $y$-range must overlap in their
interior.
Assume their $y$-range overlaps, the other case is similar.  Which polygon
should occupy the area that is between $P(j)$ and $P(l)$ horizontally?
It cannot be $k$, because $P(k)$ contains corners 1 and 2 and hence would
obtain 2 reflex angles from $P(j)$ and $P(l)$.  So it must be $P(m)$, since
$n$ is not adjacent to $j$ and $l$.
But $P(m)$ must
also separate $P(n)$ from both $P(j)$ and $P(l)$.  Regardless of whether
$n$ or $l$ occupies corner 3, this is not possible without two reflex
vertices for $m$.
Therefore either the areas are not respected or some polygon must have
8 sides.
\end{proof}

Lemma~\ref{lem:ham-eight} together with Lemma~\ref{lem:ham-counter} yield the
following theorem.

\begin{theorem}
 Eight-sided polygons are always sufficient and sometimes necessary for a cartogram of a Hamiltonian maximal planar graph.
\end{theorem}

\section{Cartograms with 6-sided Polygons}
\label{sec:six-side}

Here we study cartograms with
rectilinear 6-sided polygons. We first note that these are
easily constructed for outer-planar graphs.  Then we generalize this
technique to other maximal planar Hamiltonian graphs.

\subsection{Maximal Outer-planar Graphs}

Our algorithm from Lemma~\ref{lem:ham-eight} naturally gives drawings
of maximal outer-planar graphs that use 6-sided polygons.
Another linear-time algorithm for constructing a cartogram of a maximal
outer-planar graph with 6-sided rectilinear polygons is also described
in~\cite{ABFGKK11}, however, our construction
based on Lemma~\ref{lem:ham-eight} is much simpler.
Any
maximal outer-planar graph $G$ can be made into a maximal Hamiltonian
graph by duplicating $G$ and gluing the copies together at the outer-face
such that $G_l=G=G_r$.  (This graph has double edges, but the algorithm
in Lemma~\ref{lem:ham-eight} can handle double edges as long as one
copy is in the left and one in the right graph.)  Create the drawing
based on Lemma~\ref{lem:ham-eight} with all vertices having double the weight,
and cut it in half with a vertical line.  This gives a drawing of $G$
with 6-sided rectilinear polygons as desired.

\subsection{One-Legged Hamiltonian Cycles}

We now aim to find more maximal Hamiltonian graphs which have cartograms with 6-sided polygons.
In a Hamiltonian cycle $v_1, \ldots, v_n$, call vertex
$v_j$ {\em two-legged} if it has a neighbor $v_i^l$ in $G_l$
with $i^l<j-1$ and also a neighbor $v_i^r$ in $G_r$ with $i_r<j-1$.
Call a Hamiltonian cycle {\em one-legged} if none of its vertices
is two-legged.
In the construction from Lemma~\ref{lem:ham-eight}, the polygon of $v_j$
obtains a reflex vertex on both sides only if it has a neighbor below
$v_{j-1}$ on both sides, or in other words, if it is two-legged.
Hence we have:

\begin{lemma}\label{lem:6-sided}
  Let $G=(V,E)$ be a maximal planar graph with a one-legged Hamiltonian
  cycle and let $w:V\rightarrow R^+$ be a weight function. Then a
  cartogram with $6$-sided polygons can be computed in linear time.
\end{lemma}

It is a natural question to characterize graphs that have such
Hamiltonian cycles. Given a Hamiltonian cycle $v_1,\ldots,v_n$ we fix a
plane embedding of $G$ with outer triangle $\{v_1,v_k,v_n\}$.
The following lemma gives charecterization of graphs with such hamiltonian cycles.

\begin{lemma}\label{lem:Hamilton-property}
  Let $v_1,\ldots,v_n$ be a Hamiltonian cycle in a maximal plane graph
  $G$ with $(v_1,v_n)$ on the outer triangle. Define $w_i := v_{n-i+1}$.
  Then the following conditions are equivalent:
  \begin{enumerate}[(a)]
  \item The Hamiltonian cycle is one-legged.\label{enum:one-legged}
  \item For $i=2,\ldots,n$, edge $(v_{i-1},v_i)$ is an outer edge
of the graph $G_i$ induced by
by $v_1,v_2,\ldots,v_i$ (with the induced embedding.)
\label{enum:outer-edge}
  \item $v_{n-1}$ is an outer vertex and vertex $v_i$ has at least
    two neighbors with a larger index for
    $i=1,\ldots,n-2$.\label{enum:two-back}
  \item $w_1,\ldots,w_n$ is a canonical ordering for
    $G$.\label{enum:canonical-ordering}
 \item $G$ admits a Schnyder realizer $(\mathcal{S}_1,\mathcal{S}_2,\mathcal{S}_3)$ in which
   $w_1$, $w_2$ and $w_n$ are the roots of $\mathcal{S}_1$, $\mathcal{S}_2$ and $\mathcal{S}_3$, respectively
   and every inner vertex is a leaf in $\mathcal{S}_1$ or $\mathcal{S}_2$.\label{enum:Schnyder-wood}
  \end{enumerate}
\end{lemma}
\begin{proof}
  (\ref{enum:one-legged}) $\Longleftrightarrow$ (\ref{enum:outer-edge}):
  For $i = 2,\ldots,n$ we argue that~(\ref{enum:one-legged}) vertex $v_i$
	is not two-legged
	if and only if~(\ref{enum:outer-edge}) holds for~$i$. Indeed,
  $(v_{i-1},v_i)$ is an inner edge in $G_i$ if and only if there are
  boundary edges $(v_i,v_j)$ and $(v_i,v_k)$ with $j, k < i-1$ in
  $G_l$ and $G_r$, respectively. But then $v_i$
	is two-legged by definition.

  (\ref{enum:outer-edge}) $\Longleftrightarrow$ (\ref{enum:two-back}):
  Since $v_n$ is an outer vertex and $G_n = G$, (\ref{enum:outer-edge})
  holds for $i = n$ if and only if $v_{n-1}$ is an outer vertex. For
  $i = 2,\ldots,n-1$ we argue that~(\ref{enum:outer-edge}) holds for $i$
  if and only if~(\ref{enum:two-back}) holds for $i-1$. Let $v_i^l$,
  respectively $v_i^r$, denote the third vertex in the inner facial
  triangle containing the edge $(v_{i-1},v_i)$ in $G_l$, respectively
  $G_r$. (Both triangles exist, since $(v_{i-1},v_i)$ is an inner edge
  in $G$.) Now $(v_{i-1},v_i)$ is an inner edge in $G_i$ if and only
  if both, $v_i^l$ and $v_i^r$, have a smaller index than $v_{i-1}$,
  which in turn holds if and only if the index of \emph{every}
  neighbor of $v_{i-1}$, different from $v_i$, is smaller than $i-1$.

  (\ref{enum:two-back}) $\Longrightarrow$ (\ref{enum:canonical-ordering}):
  By~(\ref{enum:two-back}) $\{w_1,w_2,w_n\} = \{v_n,v_{n-1},v_1\}$ is
  the outer triangle of $G$, and moreover, $\tilde{G}_3$, which is
  induced by $v_n,v_{n-1},v_{n-2}$, is a triangle. Hence the outer
  boundary of $\tilde{G}_3$ is a simple cycle $C_3$ containing the
  edge $(w_1,w_2)$. In other words, the first condition of a canonical
  ordering is met for $i=4$. Assuming~(\ref{enum:two-back}) and the
  first condition for $i = 4,\ldots,n-1$, we show that the second and
  first condition holds for $i$ and $i+1$, respectively. In the end,
  the second condition holds for $i = n$ since $w_n$ is an outer
  vertex.

  First note that $w_i$ is in the exterior face of $\tilde{G}_{i-1}$
  since $w_n$ lies in the exterior face and the path $w_i,\ldots,w_n$
  is disjoint from vertices in $\tilde{G}_{i-1}$ and the embedding is
  planar. By~(\ref{enum:two-back}) $w_i$ has at least two neighbors in
  $\tilde{G}_{i-1}$. If the neighbors would not form a subinterval of the path
  $C_{i-1} \setminus (w_1,w_2)$, there would be a non-triangular inner
  face in $\tilde{G}_i$, which contains a vertex $w_j$ with $j > i$ in
  its interior. But then the path $w_j,\ldots,w_n$, which is disjoint
  from $\tilde{G}_i$, would start and end in an interior and the
  exterior face of $\tilde{G}_i$, respectively. This again contradicts
  planarity. Thus the second condition of a canonical ordering is
  satisfied for $i$. Moreover $\tilde{G}_i$ is internally
  triangulated, has a simple outer cycle $C_i$ containing the edge
  $(w_1,w_2)$. In other words, the first condition holds for $i+1$.

  (\ref{enum:canonical-ordering}) $\Longrightarrow$ (\ref{enum:two-back}):
  Since $w_1,\dots,w_n$ is a canonical ordering, $(w_1,w_2)$ is an outer edge. In
  particular, $w_2 = v_{n-1}$ is an outer vertex. Clearly $v_1$ has at
  least two neighbors and every neighbor has a larger index,
  i.e.,~(\ref{enum:two-back}) holds for $i=1$.  Moreover, by the second
  condition of a canonical ordering every vertex $v_i = w_{n-i+1}$,
  for $i=2,\ldots,n-2$, has at least two neighbors in $\tilde{G}_{n-i}
  = G \setminus G_i$, which is the subgraph induced by
  $v_n,\ldots,v_{i+1}$.

 (\ref{enum:canonical-ordering})
$\Longrightarrow$~(\ref{enum:Schnyder-wood}): Consider the Schnyder
realizer $(\mathcal{S}_1,\mathcal{S}_2,\mathcal{S}_3)$ of $G$ defined
by the canonical order $w_1,\ldots,w_n$ according to
Lemma~\ref{lemma:can-schny}. For $i=3,\ldots,n-1$ the outer cycle
$C_i$ of $\tilde{G}_i$ consists of the edge $(w_1,w_2)$, the
$w_iw_1$-path $P_1$ in $\mathcal{S}_1$, and the $w_iw_2$-path $P_2$ in
$\mathcal{S}_2$. Due to the counterclockwise order of edges in a
Schnyder realizer, no vertex on $P_1$, respectively $P_2$, has an
incoming inner edge in $\tilde{G}_i$ in $\mathcal{S}_2$, respectively
$\mathcal{S}_1$. Thus considering only edges in $\tilde{G}_i$ every
outer vertex in $\tilde{G}_i$, different from $w_1$, $w_2$, is a leaf
in $\mathcal{S}_1$ or $\mathcal{S}_2$. When in the canonical ordering
vertex $w_{i+1}$ is attached to $\tilde{G}_i$, some vertices on $C_i$
become inner vertices of $\tilde{G}_{i+1}$. Every inner edge in
$\tilde{G}_{i+1}$, which was not an edge in $\tilde{G}_i$ is in
$\mathcal{S}_3$. Thus every inner vertex in $\tilde{G}_i$ is a leaf in
either $\mathcal{S}_1$ or $\mathcal{S}_2$.

 (\ref{enum:Schnyder-wood})
$\Longrightarrow$~(\ref{enum:canonical-ordering}): Consider a
canonical ordering $w_1,w_2,\ldots,w_n$ of $G$ defined by the Schnyder
realizer $(\mathcal{S}_1,\mathcal{S}_2,\mathcal{S}_3)$ according to
Lemma~\ref{lemma:can-schny}. Then $\{w_1,w_2,w_3\}$ is a triangle,
hence $C_3$ consists of the edge $(w_1,w_2)$, the $w_3w_1$-path $P_1$
in $\mathcal{S}_1$, and the $w_3w_2$-path $P_2$ in
$\mathcal{S}_2$. For $i = 4,\ldots,n$ the vertex $w_i$ is attached to
$\tilde{G}_{i-1}$. If $w_{i+1}$ would have no edge to $w_i$ then the
outgoing edge of $w_i$ in $\mathcal{S}_1$ or $\mathcal{S}_2$ is
connected an inner vertex in the $w_iw_2$-path or $w_iw_1$-path,
respectively. But this vertex would then have an incoming edge in
both, $\mathcal{S}_1$ and $\mathcal{S}_2$ -- a contradiction.
\end{proof}

Figure~\ref{fig:Can-Ordering} shows an example of a maximal plane
graph with a one-legged Hamiltonian cycle, the corresponding canonical
ordering, and the Schnyder realizer.

\begin{figure}[htbp]
%\vspace{-0.5cm}
 \centering
 \includegraphics[width=0.7\textwidth]{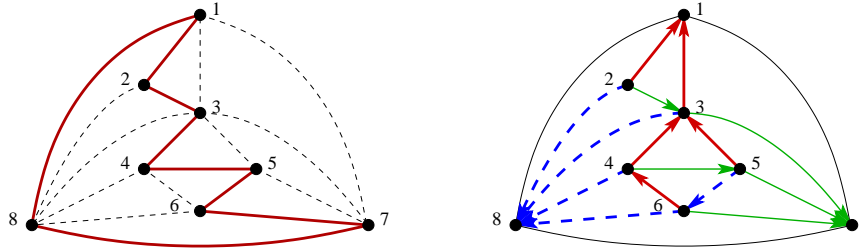}
 \caption{\small\sf A
  graph with a one-legged Hamiltonian cycle and
 the corresponding Schnyder realizer.}
 \label{fig:Can-Ordering}
\end{figure}

Once we have a one-legged Hamiltonian cycle, we can build a
6-sided cartogram via Lemma~\ref{lem:6-sided} in linear
time. Alternately we could obtain from it a Schnyder wood, rooted such
that every vertex is a leaf in $\mathcal{S}_1$ or $\mathcal{S}_2$, and
hence obtain a 6-sided cartogram via the algorithm in
Section~\ref{sec:eight-side}. However, we prefer the construction of Lemma~\ref{lem:6-sided} due to
its linear runtime.

We know that not every Hamiltonian maximally planar graph admits a
one-legged Hamiltonian cycle; for example, the graph in
Figure~\ref{fig:ham-counter} does not even admit a cartogram with
6-gons. However, we believe that some non-trivial subclasses of
Hamiltonian maximally planar graphs are also one-legged Hamiltonian. In particular, we have the following conjecture:

\begin{conjecture}
Every 4-connected maximal planar graph has a one-legged Hamiltonian cycle.
\end{conjecture}

Note that by Lemma~\ref{lem:Hamilton-property}, the conjecture is
equivalent to asking whether every 4-connected maximal graph has
a Hamiltonian cycle such that taking the vertices in this order gives
a canonical ordering.  Such a result might be of use for other graph
problems as well.

\section{Conclusion and Open Problems}
\label{sec:conc}

We presented a cartogram construction for maximal planar graphs with optimal polygonal complexity. For the precise realization of the actual cartogram this approach requires numerical iteration. Even though the simple heuristic works well in practice, a natural open problem is whether everything can be computed with an entirely combinatorial linear-time approach.

We also presented such an entirely combinatorial linear-time
construction for Hamiltonian maximal planar graphs and showed that the
resulting 8-sided cartograms are optimal. Finally, we showed that if the graph admits a one-legged Hamiltonian cycle (for example outer-planar graphs), only 6 sides are needed. It remains to identify larger classes
of planar graphs which are one-legged Hamiltonian and thus have 6-sided
cartograms. We conjecture that 4-connected maximal planar graphs have
this property.

All of the constructions in this paper yield area-universal
rectilinear duals with optimal polygonal complexity. While Eppstein
{\em et al.}~\cite{EMVS} characterized area-universal {\em rectangular}
layouts, a similar characterization remains an open problem for general area-universal
{\em rectilinear} layouts.

For some classes of graphs the unweighted and weighted versions of the
problem have the same polygonal complexity, as in the case of general
planar graphs where we have shown that the tight bound of 8-sided for
weighted graphs matches the tight bound for unweighted graphs. On the
other hand, Hamiltonian maximal planar graphs have a tight bound of 6 in the
unweighted case, while we have shown that the tight bound is 8 in the
weighted case. It would be interesting to study when the weighted
version of the problem increases the polygonal complexity.

In a similar vein, rectilinear representations are often desirable for
practical and technical reasons (e.g., for VLSI layout or
floor-planning). Sometimes, insisting on rectilinear representation
increases the underlying polygonal complexity. For example, general
(unweighted) planar graphs can be represented by 6-sided polygons (tight
bound) while 8 are needed in the rectilinear case. For the weighted
version, we also now know that 8 sided are sufficient in the
rectilinear case, but can we improve this to 7 sides if we do not
insist on rectilinear layouts?

%\newpage

\bibliographystyle{abbrv}
{
\bibliography{stephen}
}

\end{document}